\documentclass[12pt]{iopart}
\usepackage{amsthm}
\usepackage{amsmath}
\usepackage{amssymb}
\usepackage{cite} 
\usepackage[urlcolor=blue,colorlinks=true,linkcolor=blue,pdfstartview={FitH},bookmarks=false]{hyperref}

\newcommand{\ket}[1]{\ensuremath{|#1\rangle}}
\newcommand{\bra}[1]{\ensuremath{\langle#1|}}

\newcommand{\braket}[2]{\ensuremath{\langle#1|#2\rangle}}

\newcommand{\eg}{\emph{e.g.}}
\newcommand{\ie}{\emph{i.e.}}
\newcommand{\mf}{\mathbf}

\newcommand{\mc}{\mathcal}
\newcommand{\mb}{\mathbb}
\newcommand{\ot}{\otimes}

\newcommand{\idQ}{\mb{I}_{\text{Q}}}
\newcommand{\idE}{\mb{I}_{\text{E}}}

\newcommand{\p}{\scriptscriptstyle{+}}
\newcommand{\m}{\scriptscriptstyle{-}}
\newcommand{\spm}{\scriptscriptstyle{\pm}}
\newcommand{\smp}{\scriptscriptstyle{\mp}}

\newtheorem{thh}{Theorem}
\newtheorem{deff}{Definition}

\newtheorem{Rem}{Remark}

\begin{document}

\title[\it Stationary states]{Stationary states of two-level open quantum systems}

\author{Bart{\l}omiej Gardas}
\address{Institute of Physics, University of Silesia, PL-40-007 Katowice, Poland}
\ead{bartek.gardas@gmail.com}
\author{Zbigniew Pucha{\l}a}
\address{Institute of Theoretical and Applied Informatics, Polish Academy
of Sciences, Ba{\l}tycka 5, 44-100 Gliwice, Poland}
\ead{z.puchala@iitis.pl}

\begin{abstract}
A problem of finding stationary states of  open  quantum systems is addressed. 
We focus our attention on a generic type of open system: a qubit coupled to its
environment. We apply the theory of block operator matrices and find  stationary 
states of two--level open  quantum systems under certain conditions applied both
on the qubit and the surrounding.
\end{abstract}

\maketitle

\section{Introduction}

In quantum mechanics, the density operator $\rho$ of a quantum system is called a stationary state if
$[H,\rho]=0$, where $H$ is a given time-independent Hamiltonian of the system. Since $\rho$ satisfies
the Liouville-von Neumann equation $i\partial_t\rho=[H,\rho]$, it is clear that stationary states are invariant
with respect to the transformation $\rho \mapsto U_t\rho U_t^{\dagger}$, where $U_t=\exp(-\rmi Ht)$ is the
time evolution operator. In other words, stationary states do not change during the time evolution.
  
For low-dimensional closed systems, the stationary states can be obtained relatively easily~\cite{sakurai,galindo1}.
It is a common situation, that a small quantum system is immersed in other, mostly large, system called the 
environment~\cite{Leggett}. Such an open system does not evolve unitarily in time. An analysis of open quantum 
systems~\cite{davies,alicki,breuer} is much more complicated as they are a stage of a variety of physical 
phenomena~\cite{dajka2,dajka4,dajka5}. The famous decoherence process~\cite{zurek} may serve as an example. In open
quantum systems character of potentially existing stationary states is not obvious.     

There are various physical problems related to the properties of open quantum systems, which has already been addressed
and intensively discussed (see \eg, ~\cite{dajka1,dajka3}). Nevertheless, the procedure of deriving stationary states is
not one of them. The existence and properties of stationary states have significant importance in quantum information 
processing and quantum theory itself. One can pose natural questions: 

\begin{enumerate}
 \item do the stationary states exist for a given open system? 
 \item what features of a given model are responsible for existence of such states?
 \item how such states can be constructed? 
\end{enumerate}

The answers to the above questions are still incomplete. For example, it is  known that the stationary states
exist for completely positive (CP)~\cite{cp} evolution of the open system, this fact follows directly from 
Schauder fixed point theorem~\cite{Chuang}. However, this is only an existential result and so far there are 
no available methods to determine explicit form of stationary states. Furthermore, the very existence of the 
stationary states in general case is an open problem \eg, in the presence of initial system-environment correlations~\cite{NCP}.

The purpose of this paper is to propose the method of calculating the stationary states in the case of two-dimensional
open quantum system. The theory of block operator matrices~\cite{Vadim,bom,spectral} is adapted to achieve this goal. 
In particular, we use the \emph{Riccati operator equation}~\cite{RiccEq} to solve the eigenproblem for the total Hamiltonian. 
It is shown how to derive the stationary states by using the solution of the equation.
  
\section{Block operator matrix approach}

We begin with a brief review of the block operator matrices approach to the problem of decoherence in the case of a single 
qubit~\cite{mgr,gardas,gardas2}. Let $\mf{H}$ be the Hamiltonian of the total system. We will assume that it has the following form

\begin{equation}
 \label{total}
\mf{H} = H_{\text{Q}}\ot\idE+\idQ\ot H_{\text{E}} + \mf{H}_{\text{int}},
\end{equation}
where $H_{\text{Q}}$ and $H_{\text{E}}$ represent the Hamiltonian of the qubit and the environment, respectively,
while $\mf{H}_{\text{int}}$ specifies the interaction between the systems. The Hamiltonian $\mf{H}$ acts on 
the Hilbert space $\mc{H}_{\text{tot}}=\mb{C}^{2}\ot\mc{H}_{\text{E}}$, where $\mc{H}_{\text{E}}$ is the Hilbert
space (possibly infinite-dimensional) related to the environment. $\idQ$ and $\idE$ are the identity operators on 
$\mb{C}^2$ and $\mc{H}_{\text{E}}$, respectively.

Since the isomorphism $\mb{C}^2\ot\mc{H}_{\text{E}}\simeq\mc{H}_{\text{E}}\oplus\mc{H}_{\text{E}}$ holds true,
the Hamiltonian~(\ref{total}) admits the block operator matrix representation~\cite{spectral}: 

\begin{equation}
 \label{bom}
\mf{H} = 
  \begin{bmatrix}
  H_{\p} & V \\
  V^{\dagger} & H_{\m}
  \end{bmatrix}
  \quad\text{on}\quad
  \mathcal{D}(\mf{H})=\left(\mathcal{D}(H_{\p})\cap\mathcal{D}(V^{\dagger})\right)\oplus
\left(\mathcal{D}(V)\cap\mathcal{D}(H_{\m})\right).
\end{equation}
All the entries of~(\ref{bom}) are operators acting on $\mc{H}_{\text{E}}$. Moreover, the diagonal entries, \ie, $H_{\spm}$
are self-adjoint. In this paper, we will focus on the case in which $V$ is bounded, thus $V^{\dagger}$ is bounded as well; 
however, no assumption on boundedness of $H_{\spm}$ is made. Under these circumstances we have 
$\mathcal{D}(\mf{H})=\mathcal{D}(H_{\p})\oplus\mathcal{D}(H_{\m})$, where domains $\mc{D}(H_{\spm})$ are assumed to be dense in $\mc{H}_{\text{E}}$.

The generally accepted procedure to obtain the reduced time evolution of the open system, the so-called reduced dynamics, reads 

\begin{equation}
 \label{reduced}
 \rho_t  = \mbox{Tr}_{\text{E}}[\mf{U}_t\Phi(\rho_0)\mf{U}_t^{\dagger}] \equiv T_t(\rho_0). 
\end{equation}
Above, $\rho_0$ specifies the state of the open system at $t=0$. The map $\Phi$ assigns to each initial state $\rho_0$ a single
state $\Phi(\rho_0)$ of the total system. The assignment map must be chosen properly so that $T_t$ can be well-defined~\cite{reduced1,
reduced2,reduced3}. For instance, if no correlations between the systems are initially present, then $\Phi(\rho_0)=\rho_0\otimes\omega$,
for some initial state of the environment $\omega$. It is worth mentioning that, if the initial state cannot be factorized, the definition
of $\Phi$ is not accessible~\cite{korelacje,erratum}. The unitary operator $\mf{U}_t=\exp(-\rmi\mf{H}t)$ describes the time evolution of
the total system. 

The map $\mbox{Tr}_{\text{E}}(\cdot):\mathcal{T}(\mathcal{H}_{\text{E}}\oplus\mathcal{H}_{\text{E}})\rightarrow M_2(\mathbb{C})$
denotes the so-called partial trace: 

\begin{equation}
 \label{partial}
  \mbox{Tr}_{\text{E}}
  \begin{bmatrix}
   M_{11} & M_{12} \\
   M_{21} & M_{22}
   \end{bmatrix} 
   =
   \begin{pmatrix}
   \mbox{Tr}M_{11} & \mbox{Tr}M_{12} \\
   \mbox{Tr}M_{21} & \mbox{Tr}M_{22}
   \end{pmatrix}.
 \end{equation}
$\mbox{Tr}(\cdot)$ refers to the usual trace operation on $\mc{H}_{\text{E}}$,  
$\mathcal{T}(\mathcal{H}_{\text{E}}\oplus\mathcal{H}_{\text{E}})$ denotes the
Banach space of trace classes operator with the trace norm: $\|A\|_{1}=\mbox{Tr}(\sqrt{AA^{\dagger}})$, whereas
$M_2(\mathbb{C})$ is the Banach space of $2\times 2$ complex matrices. Note, the partial trace is a linear operation
transforming the block operator matrices (square brackets) to the ordinary matrices (round brackets). 

\section{Main results}
\label{main}

Fixed point theorems like Banach or Schauder indicate the existence of stationary states for a given evolution $T_t$.
However, there is no general analytical procedure to obtain explicit form of such states. In this section we propose
a method of deriving stationary states for two-level open quantum systems. The generalization to the higher-dimensions
seems to be possible. However, we will not deal with this issue in this paper. We begin with some definitions.

\begin{deff}
The density matrix $\rho$ is said to be a stationary state if it is invariant
with respect to reduced evolution, $T_t(\rho)=\rho$. 
\end{deff}
\begin{deff}
Let $X$ be an operator acting on the Hilbert space $\mc{H}_{\rm{E}}$. The subset $\Gamma_X$ of $\mc{H}_{\rm{E}}\oplus\mc{H}_{\rm{E}}$
defined as
    
   \begin{equation}
      \label{space}
      \Gamma_X:= \left\{    
         \begin{bmatrix}
              \ket{\psi}\\
              X\ket{\psi}
         \end{bmatrix} 
         :\ket{\psi}\in\mathcal{D}(X)\subset\mc{H}_{\text{E}}
              \right\}
    \end{equation} 
is said to be the graph of $X$.   
\end{deff}
The graph of a linear and closed operator is a subset of the Hilbert space, which is a Hilbert space itself equipped with the inner product

\begin{equation}
\label{inner}
\braket{\Psi_1}{\Psi_2} = \braket{\psi_1}{\psi_2}+\braket{\phi_1}{\phi_2},\quad
\ket{\Psi_i}=
\begin{bmatrix}
 \ket{\psi_i}\\
 \ket{\phi_i}
\end{bmatrix}\in\Gamma_X \quad\text{(}i=1,2\text{)}.
\end{equation}
$\braket{\psi}{\phi}$ is an inner product on $\mc{H}_{\text{E}}$. It is a known fact (see Lemma $5.3$ in~\cite{graph}) that the graph $\Gamma_X$
is $\mf{H}-$invariant, that is $\mf{H}\left(\Gamma_X\cap\mathcal{D}(\mf{H})\right)\subset\Gamma_X$ if and only if $X$ is a bounded solution (with 
$\text{Ran}(X|_{\mathcal{D}(H_{\p})})\subset\mathcal{D}(H_{\m})$) of the \emph{Riccati} equation:

\begin{equation}
 \label{ricc}
       XVX+XH_{\p}-H_{\m}X-V^{\dagger}=0\quad \text{on}\quad\mathcal{D}(H_{\p}).
\end{equation}
Along with the equation above we introduce the dual Riccati equation, namely

\begin{equation}
 \label{ricc2}
       YV^{\dagger}Y+YH_{\m}-H_{\p}Y-V=0\quad \text{on}\quad\mathcal{D}(H_{\m}).
\end{equation}
It is proved in~\cite{graph} that $Y=-X^{\dagger}$ is a solution (with $\text{Ran}(X^{\dagger}|_{\mathcal{D}(H_{\m})}\subset\mathcal{D}(H_{\p})$)
of~(\ref{ricc2}) if and only if the orthogonal complement of $\Gamma_X$, \ie, the subspace 

\begin{equation}
      \label{space2}
      \Gamma_X^{\bot}= \left\{    
         \begin{bmatrix}
              -X^{\dagger}\ket{\psi}\\
              \ket{\psi}
         \end{bmatrix} 
         :\ket{\psi}\in\mathcal{D}(X^{\dagger})\subset\mc{H}_{\text{E}}
              \right\}
    \end{equation}
is $\mf{H}$-invariant. It is straightforward to see that a bounded operator $X$ solves~(\ref{ricc}) if and only if $-X^{\dagger}$
is a solution of~(\ref{ricc2}). Therefore, $\Gamma_X$ as well as $\Gamma_X^{\bot}$ are $\mf{H}-$invariant if and only if $X$ is a
bounded solution of~(\ref{ricc}). In other words, $\Gamma_X$ is reducing subspace of $\mf{H}$ if and only if $X$ is a bounded solution
of~(\ref{ricc}). From considerations above follow also that $\Gamma_X$ and $\Gamma_X^{\bot}$ are $\mf{U}_t-$invariant.

\begin{deff}
\label{def}
Elements from the graph and its orthogonal complement are denoted by $\ket{X_{\psi}}$ and $\ket{X^{\psi}}$, respectively. The Riccati
states are defined as $\rho_{\psi}=\rm{Tr}_{\rm{E}}(\varrho_{\psi})$ and $\rho^{\psi}=\rm{Tr}_{\rm{E}}(\varrho^{\psi})$, where 
$\varrho_{\psi}:=\ket{X_{\psi}}\bra{X_{\psi}}$ and $\varrho^{\psi}:=\ket{X^{\psi}}\bra{X^{\psi}}$.                        
\end{deff}
The vectors $\ket{X_{\psi}}$ and $\ket{X^{\psi}}$ are \emph{not} normalized with respect to the norm induced by the inner
product~(\ref{inner}). Moreover, the states $\varrho_{\psi}$ and $\varrho^{\psi}$ are not factorisable (\ie, correlations occur),
unless $X\ket{\psi}\sim\ket{\psi}$ and $X^{\dagger}\ket{\psi}\sim\ket{\psi}$, respectively. However, they are 
$\mf{U}_t(\cdot)\mf{U}_t^{\dagger}-$invariant, which is obvious because the vectors $\ket{X_{\psi}}$ and $\ket{X^{\psi}}$ are
$\mf{U}_t-$invariant. As a consequence, the Riccati states $\rho_{\psi}$ and $\rho^{\psi}$ are $T_t-$invariant, where the map
$T_t$ has been defined in~(\ref{reduced}). Therefore, the set of all the Riccati states is invariant under the time evolution.
Nevertheless, the Riccati states are not the stationary states, in general. However, we show that the latter can be found among
the Riccati states. To be specific, we will prove the following

\begin{thh}
  \label{th1}
Let $X$ be a bounded solution of the Riccati equation~\eref{ricc}. Then
\begin{enumerate}
\item[i)] the Riccati state $\rho_{\psi}$ is a stationary state if the vector
          $\ket{\psi}$ is a eigenvector of the operator $Z_{\p}\equiv H_{\p}+VX :\mathcal{D}(H_{\p})\rightarrow\mc{H_{\text{E}}}$,
\item[ii)] the Riccati state $\rho^{\phi}$ is a stationary state if the vector
          $\ket{\phi}$ is a eigenvector of the operator $Z_{\m}\equiv H_{\m} - 
          V^{\dagger}X^{\dagger}:\mathcal{D}(H_{\m})\rightarrow\mc{H_{\text{E}}}$.
\end{enumerate}
\end{thh}

\begin{proof}
Let $Z_{\p}\ket{\psi}=\lambda\ket{\psi}$ for $\lambda\in\mb{C}$ and $\ket{\psi}\in\mathcal{D}(H_{\p})$.
From~(\ref{ricc}) we obtain that $V^{\dagger}+H_{\m}X=XZ_{\p}$, hence in view of~(\ref{bom}) the last equality
leads to $\mf{H}\ket{X_{\psi}} = \lambda\ket{X_{\psi}}$. Thus, the vector state $\ket{X_{\psi}}$ is the
eigenvector of the total Hamiltonian with the corresponding eigenvalue $\lambda$. Since $\mf{H}$ is self-adjoint
we have $\lambda\in\mb{R}$ and in consequence $\mf{U}_t\varrho_{\psi}\mf{U}_t^{\dagger}=\varrho_{\psi}$,
which ultimately leads to $T_t(\rho_{\psi})=\rho_{\psi}$.

In a comparable manner we have $\mf{H}\ket{X^{\phi}}=\xi\ket{X^{\phi}}$ for $\xi\in\mb{R}$ and 
$\ket{\phi}\in\mc{D}(H_{\m})$ so that $Z_{\m}\ket{\phi}=\xi\ket{\phi}$. Just as before 
$\mf{T}_t(\varrho^{\phi})=\varrho^{\phi}$, therefore $T_t(\rho^{\phi})=\rho^{\phi}$.
\end{proof}

At this point, some remarks, regarding theorem given above, should be made.   

\begin{Rem}
The question whether all stationary states are Riccati states or if it is possible that stationary states exist that
are not Riccati states is still open.
\end{Rem}
\begin{Rem}
 \label{re1}
Since the space $\Gamma_{X}$ is closed, we have the following decomposition $\mc{H}_{\text{tot}}= 
\Gamma_{X}\oplus\Gamma_{X}^{\bot}$. Thus the total Hamiltonian is similar to certain block diagonal
operator  matrix, $\mf{S}^{-1}\mf{H}\mf{S}=\mf{H}_{\text{d}}$, where

\begin{equation}
 \mf{H}_{\text{d}}= 
      \begin{bmatrix}
        Z_{\p} & 0 \\
        0 & Z_{\m}
      \end{bmatrix}
\quad\text{with}\quad\mathcal{D}(Z_{\spm})=\mathcal{D}(H_{\spm})\quad\text{and}\quad      
 \mf{S}= 
      \begin{bmatrix}
     \idE & -X^{\dagger} \\
        X & \idE
    \end{bmatrix}. 
\end{equation} 
This, implies that $\sigma(\mf{H})=\sigma(Z_{\p})\cup\sigma(Z_{\m})$. Therefore, the eigenvalues of $Z_{\spm}$ are exactly 
the eigenvalues of the Hamiltonian $\mf{H}$.
\end{Rem}

\begin{proof}
Let $X$ be a bounded solution of~(\ref{ricc}). $V$ is assumed to be bounded as well. From the definition of $Z_{\spm}$ we have 
$\mathcal{D}(Z_{\spm})=\mathcal{D}(H_{\spm})$, and thus $\mc{D}(\mf{H})=\mc{D}(\mf{H}_{\text{d}})$. Since $X$ solves the Riccati
equation~(\ref{ricc}), it is clear that $\mf{H}\mf{S}=\mf{S}\mf{H}_{\text{d}}$. To prove $\mf{H}\sim\mf{H}_{\text{d}}$ we will 
show that $\mf{S}$ is invertible and $\mf{S}^{-1}$ is bounded. Indeed, $\mf{S}=\mf{I}+\mf{X}$, where

   \begin{equation}
       \mf{X} = 
       \begin{bmatrix}
        0 & -X^{\dagger} \\
        X & 0
        \end{bmatrix}.
       \end{equation}
Since $\mf{X}^{\dagger}=-\mf{X}$, the spectrum of $\mf{X}$ is a subset of the imaginary axis. In particular, $-1\not\in\sigma(\mf{X})$ 
thus $0\not\in\sigma(\mf{S})$ and, hence, $\mf{S}$ has a bounded inverse. 
\end{proof}
\begin{Rem} 
  \label{Rem4}
The stationary states $\rho_{\psi}$, $\rho^{\phi}$ indicated in theorem~\ref{th1} are given by

 \begin{equation}
  \label{explicit}
  \rho_{\psi}=A
      \begin{pmatrix}
          1 & \langle X\rangle^{*}_{\psi} \\
           \langle X\rangle_{\psi} & \|X\psi\|^2
      \end{pmatrix}\quad\text{and}\quad
        \rho^{\phi}=B
      \begin{pmatrix}
          \|X^{\dagger}\phi\|^2 & -\langle X\rangle^{*}_{\phi} \\
          -\langle X\rangle_{\phi} & 1
      \end{pmatrix},
 \end{equation}
where $\ket{\psi}\in\mathcal{D}(H_{\p})$ and $\ket{\phi}\in\mathcal{D}(H_{\m})$ are the eigenvectors of $Z_{\m}$ and $Z_{\p}$,
respectively. $A=\mbox{\emph{Tr}}(\rho_{\psi})$, $B=\mbox{\emph{Tr}}(\rho^{\phi})$ are the normalization constants and $\langle 
X\rangle_{\varphi}=\bra{\varphi}X\ket{\varphi}$.
\end{Rem}
\begin{proof}
Since $\mbox{Tr}\ket{\psi}\bra{\phi}=\braket{\phi}{\psi}$, the equations~(\ref{explicit}) can be obtained directly from the 
definition~(\ref{def}) and the formula~(\ref{partial}). 
\end{proof}
%
\section{Examples}
 
\subsection{Spin-boson model}
%
In this subsection we will demonstrate an application of the presented method to a non-trivial example, namely, the paradigmatic
spin-boson model~\cite{SB_paradigm, SB_Parity}. Assume that the Hamiltonian of the qubit (spin-half) and its environment (boson) 
are in the following forms

\begin{equation}
 \label{eq:QandE}
    H_Q = \beta\sigma_z+\alpha\sigma_x \quad\text{and}\quad
    H_E = \omega a^{\dagger}a,      
\end{equation}
respectively. For the sake of simplicity, we consider only the case of a single boson. The interaction between the systems reads

\begin{equation}
 \mf{H}_{\text{int}} = \sigma_z\ot (g^{\ast}a+ga^{\dagger})
                       \equiv\sigma_z\ot V.
\end{equation}
In the above description, $\sigma_x$, $\sigma_z$ are the standard Pauli matrices and $\alpha, \beta \in \mb{R}$. The creation 
$a^{\dagger}$ and annihilation $a$ operators obey the canonical commutation relation (CCR) $[a,a^{\dagger}]=\idE$~\cite{sakurai}.
Parameters $\omega>0$ and $g\in\mathbb{C}$ represent the energy of the boson and the coupling constant between the qubit and the
environment, respectively.

If $\alpha=0$ (no energy exchange between the systems), the model can be solved, \ie, the reduced dynamics can be obtained
exactly~\cite{SB_Alicki,Luczka}. The solution describes the physical phenomena known as the pure decoherence or dephasing%
~\cite{dajka_cat}. On the other hand, when $\alpha\not=0$ the exact solution in not known. The objective is to estimate the
stationary states for the latter case. 

To proceed, we must clarify some technical aspects (\eg, domains of $H_{\spm}$). Clearly, the operators $a$ and $a^{\dagger}$ cannot
both be bounded since the trace of their commutator does not vanish~\cite{reed}. Therefore, the CCR holds only on some dense subspace
$\mc{D}_2$ of $\mc{H}_{\text{E}}$. Let $\mc{D}_1$ be the dense domain of both $a$ and $a^{\dagger}$, on which they are mutually adjoint,
that is, $(a^{\dagger})^{*}=a$, $a^{*}=a^{\dagger}$. At this point, it is not obvious that the sets $\mc{D}_1$, $\mc{D}_2$, having desire
properties, exist. The detailed construction can be found in~\cite{sz0} and the right choice is given by

\begin{equation}
 \mc{D}_k = \left\{\ket{\psi}\in\mc{H}_{\text{E}}: \sum_{n=0}^{\infty}n^k|\braket{\psi}{\phi_n}|^2<\infty\right\},\quad k =1,2.
\end{equation}
On $\mc{D}_1$ the creation and annihilation operators can be defined explicitly as (see also~\cite{sz1,sz2,sz3})

\begin{equation}
 \label{aa}
   a\ket{\phi} = \sum_{n=1}^{\infty}\sqrt{n}\braket{\phi_n}{\phi}\ket{\phi_{n-1}}, \quad
   a^{\dagger}\ket{\phi} = \sum_{n=0}^{\infty}\sqrt{n+1}\braket{\phi_n}{\phi}\ket{\phi_{n+1}},
   \quad\ket{\phi}\in\mc{D}_1.
  \end{equation}
$\{\ket{\phi_n}\}_{n=0}^{\infty}$ is an orthonormal basis in $\mc{H}_{E}$. From~(\ref{aa}) follows that
$a^{\dagger}\ket{\phi_n}=\sqrt{n+1}\ket{\phi_{n+1}}$ and $a\ket{\phi_n}=\sqrt{n}\ket{\phi_{n-1}}$, which in most books
on quantum mechanics is a definition of the creation and annihilation operators. However, the operators defined
in such a way are not closed; nevertheless, they are closable and their closures are given by~(\ref{aa}).

Since $a$ is closed, $H_E$ is a positive self-adjoint operator and $\mc{D}_2\subset\mc{D}_1$ is a core of $a$ (see \eg, Theorem
4.2.1 in~\cite{blank}). Henceforward, we assume that the basis $\{\ket{\phi_n}\}_{n=0}^{\infty}$ is composed with the eigenvectors
of $H_E$. In this case we have

\begin{equation}
H_E\ket{\phi} = \sum_{n=0}^{\infty}\omega n\braket{\phi_n}{\phi}\ket{\phi_n}, \quad\ket{\phi}\in\mc{D}_2.
\end{equation}
By choosing a suitable coupling constant $g$, it is possible to make $H_{\spm} = H_E \pm V$ self-adjoint on $D_{\spm}:=\mc{D}_1\cap\mc{D}_2=\mc{D}_2$.
To see this, let us first note that $(g^*a+ga^{\dagger})^{*}\supset g^*a+ga^{\dagger}$, thus $V$ is Hermitian (symmetric). Since $H_E$ is
self-adjoint and $V$ Hermitian, it is sufficient to show that $V$ is relatively bounded with respect to $H_E$ ($H_E$ bounded) and with 
$H_E$-bound less than one. The fundamental result of perturbation theory, known as the Kato-Rellich theorem~\cite{reed} assure Hermiticity
of $H_{\spm}$ in this case. 

Recall that $B$ is $A$ bounded if i) $\mc{D}(A)\subset\mc{D}(B)$ and ii) $\|B\ket{\phi}\|^2\le a\|A\ket{\phi}\|^2 + 
b\|\ket{\phi}\|^2$, for all $\ket{\phi}\in\mc{D}(A)$ and some nonnegative constants $a$, $b$. The infimum of all $a$ for which a corresponding
$b$ exists such that the last inequality holds is called $A-$bound of $B$. Note, sometimes it is convenient to replace the condition ii) by
the equivalent one: $\|B\ket{\phi}\|\le a\|A\ket{\phi}\| + b\|\ket{\phi}\|$. 

It is not difficult to see that if two operators $B_i$, $i=1,2$ are bounded with respect to the same operator $A$ and their relative
bound are less than $b_1$ and $b_2$, respectively; then $a_1B_1+a_2B_2$ is also $A$ bounded and its relative bounded is less than 
$|a_1|b_1+|a_2|b_2$ (see Lemma 6.1 in~\cite{gt}). In other words, the set of all $A$ bounded operator form a linear space. Therefore,
to see that $V$ is $H_E$ bounded it is sufficient to prove that both $a$ and $a^{\dagger}$ are $H_E$ bounded. To finish this note

\begin{eqnarray}
 \|a\ket{\phi}\|^2 &=& \sum_{n=1}^{\infty}n|\braket{\phi_{n-1}}{\phi}|^2 = \sum_{n=0}^{\infty}(n+1)|\braket{\phi_n}{\phi}|^2 \\\nonumber
 &\le& \sum_{n=0}^{\infty}n^2|\braket{\phi_n}{\phi}|^2 +\sum_{n=0}^{\infty}|\braket{\phi_n}{\phi}|^2 
  = \omega^{-1}\|H_E\ket{\phi}\|^2 + \|\ket{\phi}\|^2.
\end{eqnarray}
In comparable manner one can also verify that $\|a^{\dagger}\ket{\phi}\|^2\le \omega^{-1}\|H_E\ket{\phi}\|^2$. Since $H_E-$bound of both 
$a$ and $a^{\dagger}$ is less than one, the $H_E-$bound of $V$ is also less than one for $|z|<1/2$.

The block operator matrix representation of the spin-boson Hamiltonian reads

 \begin{equation}
  \label{eq:SBbom}
 \mf{H} = 
   \begin{bmatrix}
    H_{\p} & \alpha \\
    \alpha & H_{\m} 
   \end{bmatrix},
    \quad\text{where}\quad H_{\spm} = H_E\pm V
    \quad\text{and}\quad\mc{D}(\mf{H})=\mc{D}_2\oplus\mc{D}_2;
\end{equation}
the quantity $\alpha$ is understood as $\alpha\idE$. For the sake of simplicity we have set $\beta=0$, the example remains 
non-trivial because $[H_{Q}\ot\idE,\mf{H}_{\text{int}}]\not=0$. The corresponding Riccati equation takes the form

\begin{equation}
 \label{eq:SBriccati} 
\alpha X^2 + XH_{\p}-H_{\m}X-\alpha =0\quad\text{on}\quad\mc{D}_2.
\end{equation}
In order to solve this equation we define an operator $P$ as

\begin{equation}
 \label{parity}
 P\ket{\psi} = \sum_{n=0}^{\infty}e^{i\pi n}\braket{\phi_n}{\psi}\ket{\phi_n},\quad\ket{\psi}\in\mc{H}_{\text{E}}.
\end{equation}
Directly from~(\ref{parity}) we have $P=P^{\dagger}$ and $P^2=\idE$, thus $P$ is both self-adjoint and unitary. Formally, $P$
can be written as $P=\exp(i\pi H_{E})$, however, unlike $H_{E}$, $P$ is everywhere defined. The Hellinger-Toeplitz theorem 
guaranties that $P$ is bounded, which can also be seen directly. Indeed, form unitarity we obtain $\|P\psi\|=\|\psi\|$, for
$\ket{\psi}\in\mc{H}_{\text{E}}$, hence $\|P\|=1$. $P$ is, in fact, the bosonic parity operator~\cite{Bender}. We will show
that $X=P$ solves~(\ref{eq:SBriccati}). Since $P^2=\idE$ it is sufficient to show that 

\begin{equation}
 \label{syl}
 PH_{\m}-H_{\p}P = 0
  \quad\text{or equivalently}\quad 
  PH_{\spm}P=H_{\smp}\quad\text{on}\quad\mc{D}_2.
\end{equation}
In order to prove~(\ref{syl}) let us first note that $P\ket{\psi}\in\mc{D}_2$, for $\ket{\psi}\in\mc{D}_2$, which means
$\mbox{Ran}(P|_{\mc{D}_2})\subset\mc{D}_2$. This follows from $|e^{i\pi n}|=1$. Furthermore, $PH_EP=H_E$ and $PVP=-V$. 
The first equality is obvious, while the second one follows from $PaP=-a$ and $Pa^{\dagger}P=-a^{\dagger}$. And as a result we
obtain~(\ref{syl}).

The stationary states in this example can be written as 

 \begin{equation}
  \label{eq:SBstate}
\rho_{\spm} = \frac{1}{2}
   \begin{pmatrix}
    1 & r_{\spm} \\
     r^{*}_{\spm} & 1 
   \end{pmatrix},
   \quad\text{where}\quad 
   r_{\spm}=\bra{\psi_{\spm}}P\ket{\psi_{\spm}}
\end{equation}
and $\ket{\psi_{\spm}}$ are the eigenvectors of $Z_{\spm}=H_{\spm}\pm\alpha P$. Unfortunately, the solution of this eigenproblem is not known
for $\alpha\not=0$. However, one can determine certain bounds on $r_{\spm}$ using properties of $P$ and $\rho_{\spm}$. First, $r_{\pm}$ are 
real numbers because $P$ is self-adjoint. From the non-negativity of $\rho_{\spm}$ we obtain that $r_{\spm}\in[-1,1]$.

Equation~(\ref{eq:SBstate}) provides estimation of the stationary states of the spin-half immersed within the bosonic bath. This result has been
obtained without any approximations. Of course, $r_{\spm}$ can be computed approximately with the use of known methods. It is important to stress
that to obtain the exact reduced dynamics for the model in question one needs to resolve an eigenvalue problem for $Z_{\spm}$. 

\subsection{Commuting environment}
In the second example we consider the Hamiltonian in the following form

\begin{equation}
 \label{spin}
   \mf{H} = \alpha\sigma_x\ot\idE+\idQ\ot H_0+\sigma_z\ot H_1, \quad\alpha\not=0.
\end{equation}
Here we assume that linear, bounded operators $H_0$, $H_1$ commute. In this example we impose restriction to the spectra of $H_0$ and $H_1$, \ie{} 
$\sigma(H_0)$, $\sigma(H_1)$ are discrete and non-degenerated. The Hamiltonian~(\ref{spin}) describes a qubit in contact with an environment and
in the presence of the magnetic field $\vec{B}=B\hat{e}_{x}$, where $B\sim\alpha$. Examples of such systems occur in the literature, 
\eg~\cite{gardas3,spinStar1,spinStar2}. The block operator matrix representation of~(\ref{spin}) is given by~(\ref{eq:SBbom}) with
$H_{\spm}=H_0\pm H_1$ and $\mc{D}(\mf{H})=\mc{H}_{\text{E}}$. The corresponding Riccati equation reads~(\ref{eq:SBriccati}). 

Using the fact that $H_0$ and $H_1$ commute, so they have a common set of eigenvectors, we write

\begin{equation}
H_0\ket{\phi_n}=\lambda_n\ket{\phi_n}\quad\text{and}\quad H_1\ket{\phi_n}=\xi_n\ket{\phi_n},
\end{equation}
where $\lambda_n\in\sigma(H_0)$, $\xi_n\in\sigma(H_1)$ and $\braket{\phi_n}{\phi_m}=\delta_{nm}$, for $n,m\in\mb{N}$.
The Riccati equation has a positive and self-adjoint solution $X=f(H_1)$, where the function $f$ is given by

\begin{equation}
 f(x)=\frac{\sqrt{x^2+\alpha^2}-x}{\alpha},\quad\text{for}\quad x\in\sigma(H_1).
\end{equation} 
Unlike to the spin-boson model, in this case the eigenproblem for $Z_{\spm}$ can be readily solved. Indeed, we have

\begin{equation}
Z_{\spm}\ket{\phi_n} = \left[\lambda_n\pm\xi_n f(\xi_n)\right]\ket{\phi_n}.
\end{equation}
According to the Remark~\ref{Rem4} we obtain

\begin{equation}
  \label{final}
 \rho_n^{\p} = C_n
   \begin{pmatrix}
   1 & f(\xi_n) \\
   f(\xi_n)^{\ast} & |f(\xi_n)|^2
   \end{pmatrix}
   \quad\text{and}\quad
   \rho_n^{\m} = C_n
   \begin{pmatrix}
   |f(\xi_n)|^2 & -f(\xi_n)^{\ast} \\
   -f(\xi_n) & 1
   \end{pmatrix},
\end{equation}
where $ C_n=(1+|f(\xi_n)|^2)^{-1}$.
In this case there are no initial correlations between the systems since $X\ket{\phi_n}\sim\ket{\phi_n}$. We wish to emphasize that there
may exist other solutions of the Riccati equation. For an explicit example see Ref.~\cite{gardas2}.
%

\subsection{Sylvester equation}
%
In the case $V=0$, the Hamiltonian~(\ref{bom}) is already in a block diagonal form. One can notice that the Riccati equation simplifies to
the Sylvester equation:

\begin{equation}
 XH_{\p}-H_{\m}X = 0\quad\text{on}\quad\mc{D}(H_{\p}).
\end{equation}
There exist at least one solution, namely $X=0$. The corresponding stationary states are given by the projections 
$P_0=\mbox{diag}(0,1)=\ket{0}\bra{0}$ and $P_1=\mbox{diag}(1,0)=\ket{1}\bra{1}$. 

\subsection{All unbounded entries}

In the last example we consider an interesting example in which all the entries of $\mf{H}$ are unbounded, but still the solution of the 
Riccati equation exists as a bounded operator. To see this, let us choose $H_{\spm}=H_0$ with domain $\mathcal{D}(H_0)$ and let us assume
that $V$ is self-adjoint with domain $\mathcal{D}(V)$. Then, the Riccati equation 

\begin{equation}
  \label{cap}
  XVX+XH_0-H_0X-V = 0
  \quad\text{on}\quad
  \mathcal{D}(H_0)\cap\mathcal{D}(V)
\end{equation}
has at least two bounded solutions, that is $X_{\spm}=\pm\idE$. The stationary states read

\begin{equation}
 \rho_{\spm} = 
 \frac{1}{2}
      \begin{pmatrix}
       1 & \varepsilon_{\spm} \\
       \varepsilon_{\spm} & 1
      \end{pmatrix},\quad\text{where}\quad\varepsilon_{\spm}=\pm 1.
\end{equation}
An example in which all the entries of a block operator matrix are bounded and the solution of the Riccati equation is unbounded has been
provided in~\cite{Vadim}.
\section{Summary.}
In this paper we have proposed the method of calculating the stationary states for two-level open quantum system. The theory of block operator
matrices has been adapted to achieve this purpose. In particular, we have related the solution of the algebraic Riccati equation to stationary
states. In the presented method, the stationary states are generated from the stationary states of the total system by tracing out the environment. 
Our investigation includes the case when the initial system-environment correlations occur. In fact, this case is embedded in the method since the
eigenstates of the total Hamiltonian are entanglement.

At the end, we want to stress that the method cannot be used when the total Hamiltonian is not known. Such a situation arises \eg, when the details
of the interaction between the systems are not accessible. We hope that despite aforementioned weaknesses the results of the paper may serve as a 
starting point for further investigations.

\ack
The authors are deeply grateful to the referee for very careful reading 
of the original manuscript and valuable suggestions.
B. Gardas would like to thank Jerzy Dajka for his suggestions and comments. 
This work was financed from the Polish science budget resources in the years
2010-2013 as a research projects: grant number N N519 442339 and project number
IP 2010 0334 70.%
\section*{References}

%
%

%
\end{document}